\documentclass[journal,twoside,web]{ieeecolor}
\usepackage{generic}
\usepackage{cite}
\usepackage{amsmath,amssymb,amsfonts}
\usepackage{graphicx}
\usepackage{algorithm,algorithmic}
\usepackage{hyperref}
\hypersetup{hidelinks=true}
\usepackage{textcomp}
\usepackage{lipsum}
\usepackage{dsfont}

\def\BibTeX{{\rm B\kern-.05em{\sc i\kern-.025em b}\kern-.08em
    T\kern-.1667em\lower.7ex\hbox{E}\kern-.125emX}}
\newtheorem{secthm}{Theorem}[section]
\newtheorem{seccor}[secthm]{Corollary}
\newtheorem{seclem}[secthm]{Lemma}
\newtheorem{secex}[secthm]{Example}

\newtheorem{secprob}[secthm]{Problem}
\newtheorem{secdefn}[secthm]{Definition}
\newtheorem{secrem}[secthm]{Remark}
\newtheorem{secasm}[secthm]{Assumption}

\newcommand{\bE} { {\mathbb E}}
\newcommand{\bP} { {\mathbb P}}
\newcommand{\bR} { {\mathbb R}}
\newcommand{\bZ} { {\mathbb Z}}

\newcommand{\cA} { {\mathcal A}}
\newcommand{\cF} { {\mathcal F}}
\newcommand{\cQ} { {\mathcal Q}}
\newcommand{\cS} { {\mathcal S}}

\def\red{\hfill $\lhd$}

\begin{document}

\title{\LARGE \bf
Design of Stochastic Quantizers for Privacy Preservation
}

\author{
    Le Liu, Yu Kawano, Ming Cao
    \\
    \thanks{The work of Liu and Cao was supported in part by the European Research Council (ERC-CoG-771687). The work of Yu Kawano was supported in part by JSPS KAKENHI Grant Number JP22KK0155.}
\thanks{Le Liu and Ming Cao are with the Faculty of Science and Engineering, Univerisity of Groningen, 9747 AG Groningen, The Netherlands {\tt \small \{le.liu, m.cao\}@rug.nl}}
    \thanks{Yu Kawano is with the Graduate School of Advance Science and Engineering, Hiroshima Univeristy, Higashi-hiroshima 739-8527, Japan 
        {\tt\small ykawano@hiroshima-u.ac.jp}}
}

\maketitle

\begin{abstract}
    In this paper, we examine the role of stochastic quantizers for privacy preservation. We first employ a static stochastic quantizer and investigate its corresponding privacy-preserving properties. Specifically, we demonstrate that a sufficiently large quantization step guarantees $(0, \delta)$ differential privacy. Additionally, the degradation of control performance caused by quantization is evaluated as the tracking error of output regulation. These two analyses characterize the trade-off between privacy and control performance, determined by the quantization step. This insight enables us to use quantization intentionally as a means to achieve the seemingly conflicting two goals of maintaining control performance and preserving privacy at the same time; towards this end, we further investigate a dynamic stochastic quantizer. Under a stability assumption, the dynamic stochastic quantizer can enhance privacy, more than the static one, while achieving the same control performance. We further handle the unstable case by additionally applying input Gaussian noise.
\end{abstract}


\section{Introduction}
Data sharing is a fundamental feature of the Internet of Things, which greatly enhances the efficiency of modern society. However, concerns are growing on how private information is not sufficiently guarded against malicious users \cite{wang2016defending}. For control engineers, privacy-preserving control, especially for control implemented in networks, has attracted growing attention \cite{tourani2017security}. In networked control systems, a local system often communicates with others and the communication leads to risks of privacy leakage\cite{Jorge2016, mo2016privacy}.  To address these concerns, various mechanisms have been proposed to protect users' privacy in networked control systems \cite{hawkins2022differentially, yazdani2022differentially,zhang2021privacy}. Note inherent to digital communication networks,  the signals transmitted are quantized, sometimes even coarsely quantized due to communication capacity, resulting in a control performance degradation \cite{Fagnani2004,Liang2019, Bikas2020}. Also note that quantization itself introduces noise into the system, which can improve privacy in the sense that the true value differs from what is transmitted. \cite{murguia2018privacy, Yu2021quantization}. This motivates us to study privacy, control performances and corresponding trade-offs under quantization in networked control systems.
 In this paper, our objective is to investigate how  quantizers preserve privacy while maintaining an acceptable control performance.


\smallskip

\subsubsection*{Literature review}
In systems and control, deterministic quantizers are typically implemented in different forms, such as uniform quantizers \cite{bu2020model,li2020distributed}, logarithmic quantizers \cite{li2020distributed, Fu2015}, and dynamic quantizers \cite{Brockett_quantizer, liberzon2003hybrid, Liberzon2006}. However, privacy criteria such as differential privacy~\cite{dwork2006differential, Asoodeh2021, zhao2022survey} have been primarily introduced in stochastic settings. Partly due to this difference in problem setting, there are few works focusing on privacy analysis under quantization, and we have made some earlier attempts to address this gap \cite{Yu2021quantization}\cite{liu2023privacy}. The paper \cite{wang2022quantization} studies privacy-preserving distributed optimization when a quantizer is designed to generate only ternary data randomly.


\smallskip

\subsubsection*{Contribution}
In this paper, we investigate the relationship between stochastic quantizers and their performance of privacy and control in discrete-time linear time-invariant systems. For privacy performances, we deal with the initial state privacy as in~\cite{Yu2021quantization, Duan2015, nozari2017differentially, altafini2019dynamical, wang2019privacy, Yu2020, wang2023differential}. As a control objective, we consider tracking control. We start our analysis from a static stochastic quantizer and show that a sufficiently large quantization step guarantees $(0,\delta)$ differential privacy for a fixed finite time instant. If a system matrix is Schur stable, we can further derive a sufficient condition for differential privacy in an infinite time. Additionally, we estimate an upper bound on tracking error caused by the quantization. The combination of differential privacy analysis and the error estimation provides a trade-off between privacy and control performance, determined by a quantization step.

In order to improve the trade-off, we then employ a dynamic stochastic quantizer and demonstrate the improvement when the system matrix is Schur stable. In particular, the dynamic stochastic quantizer can make a mechanism more private than the static one while maintaining the same control performance by selecting the initial and terminal quantization steps appropriately. Finally, we also address the case where a system matrix is unstable by additionally applying input Gaussian noise and again show that a dynamic stochastic quantizer gives a better trade-off performance than a static stochastic one.

The main contributions are summarized as follows.
\begin{enumerate}
    \item We provide a differential privacy condition for a static stochastic quantizer, demonstrating a larger quantization step makes a mechanism more private.
    \item We estimate the tracking error caused by quantization for a static stochastic quantizer. A smaller quantization step results in a smaller error. Thus, a trade-off between privacy and control performances is characterized by a quantization step.
    \item When a system matrix is Schur stable, we improve the trade-off by utilizing a dynamic stochastic quantizer. This enables us to preserve the privacy while maintaining control performance.
    \item When a system matrix is unstable, we develop a new privacy-preserving mechanism with additional help of input Gaussian noise and show that a dynamic stochastic quantizer gives a better trade-off performance than a static stochastic quantizer again.
\end{enumerate}


\smallskip

\subsubsection*{Organization}
The remainder of this paper is organized as follows. In Section~\ref{sec:pre}, we provide a problem formulation and illustrate the utility of stochastic quatnizers against deterministic ones. In Section~\ref{sec:static}, we derive a sufficient condition for differential privacy and estimate an upper bound on tracking error for a static stochastic quantizer. Subsequently, we demonstrate a trade-off between privacy and control performances. In Section~\ref{sec:dsq}, we demonstrate the utility of a dynamic stochastic quantizer to improve the trade-off when a system matrix is Schur stable. In Section~\ref{sec:ust}, we consider the unstable case by additionally applying input Gaussian noise. Section~\ref{sec:sim} provides a numerical example to illustrate the proposed results. Finally, Section~\ref{sec:con} concludes the paper. All proofs are given in the Appendix.


\smallskip

\subsubsection*{Notation}
The sets of real numbers and non-negative integers are denoted by $\bR$ and $\bZ_{+}$, respectively. The absolute value of a real number is denoted by $|\cdot|$. The vector $i$-norm or induced matrix $i$-norm is denoted by $| \cdot |_i$, $i=1,2$. For $P \in \bR^{n \times n}$, $P \succ 0$ (resp. $P \succeq 0$) means that $P$ is symmetric and positive (resp. semi) definite. A probability space is denoted by $(\Omega, \cF, \bP )$, where $\Omega$, $\cF$, and $\bP $ denote the sample space, $\sigma$-algebra, and probability measure, respectively.



\section{Preliminaries}\label{sec:pre}

In this paper, our interest is to protect, as private information, the initial state of a system from being inferred when commissioning a fusion center to design control signals for achieving tracking control. For tracking control, the steady-state can be public information especially when the reference signal is public/eavesdropped. In contrast, transient behavior determined system's properties can be protected, which can be formulated as privacy protection of the initial state.

Figure \ref{fig:diagram} shows the architecture of the considered mechanism, consisting of a local system, a remote fusion center, and communication networks. The local system sends its measurements to the remote center after quantization. Then, the fusion center returns a control input signal. Local information in the green dash box is not eavesdropped while information of communication networks and the fusion center including reference signal in the red dash box can be eavesdropped. When information is communicated through networks, they are quantized due to communication capacity. It is well known that quantization degenerates control performance. In contrast, it can increase privacy performance. In this paper, we investigate the role of quantizers in privacy protection. In particular, we show the utility of static/dynamic stochastic quantizers to balance control and privacy performance.

\begin{figure}
    \centering    \includegraphics[width = 0.5\textwidth]{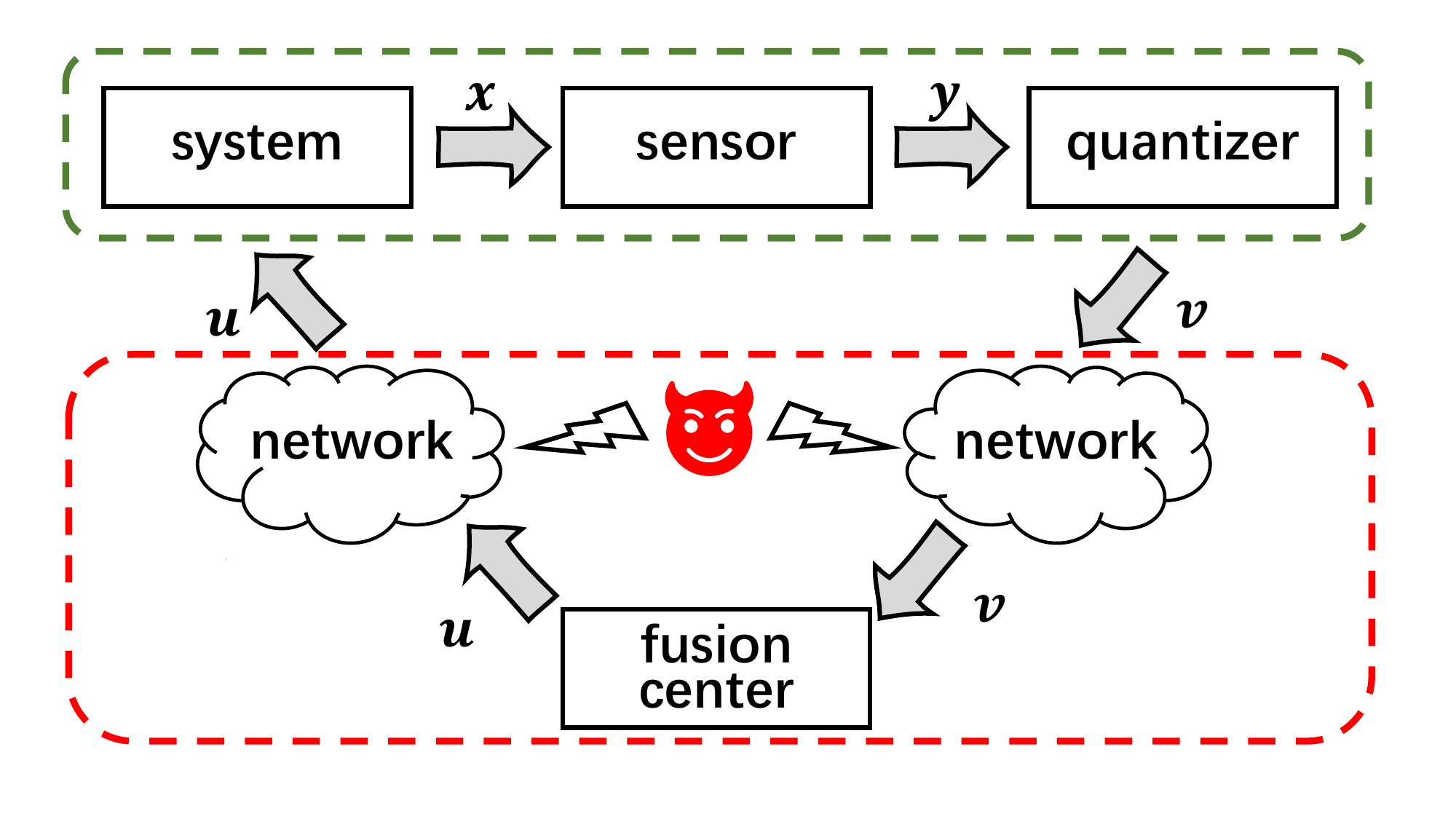}
    \caption{Mechanism diagram: The system sends its quantized information to the fusion center, and then the fusion center returns a control input to the system through communication networks, where local information in the green dash box is not eavesdropped while information in the red dash box can be.}
    \label{fig:diagram}
\end{figure}


\subsection{Problem Formulation}
The control objective in this paper is tracking control.
The system model including the sensor is given by
\begin{subequations}\label{eq:sys}
 \begin{align}
 \label{system_1}
	x(k+1) &= Ax(k) + B u(k), \quad x(0) = x_0 \\
 \label{system_2}
    y(k) &= C x(k),\\
\label{system_3}
	y_p(k) &= H_p x(k),
 \end{align}
\end{subequations}
where $x \in \bR^{n_1}$, $u \in \bR^m$, $y \in \bR^p$, and $y_p \in \bR^{q}$ denote the state,  control input, measured output, and tracking output, respectively.

The reference signal to be tracked is generated by the following exo-system:
\begin{subequations}\label{reference_sys}
\begin{align}
\label{reference dynamics}
    x_r(k+1) &= A_r x_r(k), \quad x_r(0) = x_{r,0}\\
    \label{reference_output}
    y_r(k) &= H_r x_r(k),
\end{align}
\end{subequations}
where $x_r \in \bR^{n_2}$ and $y_r \in \bR^q$. The control objective is 
\begin{align}\label{eq:conob}
\lim_{k \rightarrow \infty} e_y (k) = 0, \quad  e_y (k) := y_p(k) - y_r(k). 
\end{align}

Differently from the standard output regulation problem~\cite{huang2004nonlinear}, we assume that the state $x_r$ of the exo-system is directly available and is not needed to be estimated. In this case, the tracking problem in~\eqref{eq:conob} can be solved even if one relaxes the detectability assumption; for more details, see Remark~\ref{rem:detectability} below. Namely, the tracking problem can be solved under the following assumptions.

\begin{secasm}
\label{asm:1}
$A_r$ has no eigenvalues with modulus smaller than 1.
\red
\end{secasm}
\begin{secasm}
\label{asm:2}
The pair $(A, B)$ is stabilizable.
\red
\end{secasm}
\begin{secasm}
\label{asm:3}
The pair
$(C, A)$ is detectable.
\red
\end{secasm}
\begin{secasm}
\label{asm:4}
The following two equations:
$$
\begin{aligned}
X A_r & =A X+B U, \\
0 & =H_p X-H_r
\end{aligned}
$$
have a pair of solutions $X \in \bR^{n_1 \times n_2}$ and $U \in \bR^{m \times n_2}$.
\red
\end{secasm}

\begin{secrem}\label{rem:detectability}
In the standard output regulation problem~\cite{huang2004nonlinear}, the detecability of the following pair is assumed:
\begin{align*}
\left( \begin{bmatrix}H_p & - H_r \end{bmatrix}, \begin{bmatrix}A_p & 0 \\ 0 & A_r \end{bmatrix} \right).
\end{align*}
However, if $x_r$ is directly available, Assumption \ref{asm:3} is enough to guarantee the solvability of the output regulation problem as shown in~\cite{Yu2020}. 
\red
\end{secrem}

According to~\cite{Yu2020}, the tracking problem in~\eqref{eq:conob} can be solved by the following dynamic controller (when no signal is quantized, i.e., $v = y$): 
\begin{align}\label{eq:fusion}
 \hat{x}(k+1) &= (A + LC + BK_x)\hat{x}(k) + BK_r x_r(k) - Lv(k) \nonumber\\
	  &=  A \hat{x}(k) + B u(k) + L (C \hat{x}(k) - v(k) ), \\
   u(k) & = K_x \hat{x}(k) + K_r x_r(k), \nonumber
\end{align}
where $\hat x \in \bR^{n_1}$ denotes the state of the controller dynamics. This controller solves the tracking problem if tuning parameters $L$, $K_x$, and $K_r$ are selected such that $A + B K_x$ and $A + L C$ are exponentially stable and $K_r = U - K_x X$ for $X$ and $U$ in Assumption~\ref{asm:4}~\cite{Yu2020}. 

The fusion center is the system~\eqref{eq:fusion} with the quantized measurement $v$ generated by
\begin{align}\label{eq:Q}
v &= \mathcal{Q}_{v}(y)
:=\begin{bmatrix}
\mathcal{Q}_{v}(y_1) \\ \vdots \\ \mathcal{Q}_{v}(y_p)
\end{bmatrix},
\end{align}
where $\mathcal{Q}_{v}$ is a quantizer subject to communication capacity constraints. There are freedoms for its design, including the selection of quantizers such as static/dynamic and deterministic/stochastic. When a stochastic quantizer is implemented, this is independent and identically distributed (i.i.d.) with respect to the vector components.

Due to quantization, tracking performance is degenerated, i.e., it can cause tracking error. In contrast, it becomes difficult to infer the actual output $y$ in transient for potential eavesdroppers who may access to the fusion center and communication networks. Accordingly, private information, i.e., the initial state $x_0$ of the system can be protected against eavesdroppers. 
In this paper, we are interested in the trade-off between control and privacy performance, stated below. The formal definition of privacy performance will be introduced later.

\begin{secprob}
\label{prob:1}
Given a system~\eqref{eq:sys}, exo-system~\eqref{reference_sys}, and fusion center~\eqref{eq:fusion}, design a quantizer~\eqref{eq:Q} to balance control and privacy performances, where the control objective is to track a reference signal as in~\eqref{eq:conob} and privacy objective is to protect the initial state $x_0$ from being inferred from $v$ and $u$.
\red
\end{secprob}

\begin{secrem}
In this paper, we do not assumes that the input signal $u(k)$ is quantized because $u(k)$ can be eavesdropped. Namely, quantization of $u(k)$ only affects to control performance analysis, which is not relevant with our main interest: the trade-off between control and privacy performances.
    \red
\end{secrem}


\subsection{A Motivating Example for Employing Stochastic Quantziers}
As mentioned above, our interest is to design a quantizer for balancing control and privacy performances. It is well known that static deterministic quantizers sometimes lead to significant degeneration of control performance \cite{fagnani2003stability,liverani1995decay,buzzi1997intrinsic} while this issue can be resolved by using stochastic quantizers. We first confirm this by the following example.

\begin{figure}
    \centering    \includegraphics[width = 0.5\textwidth]{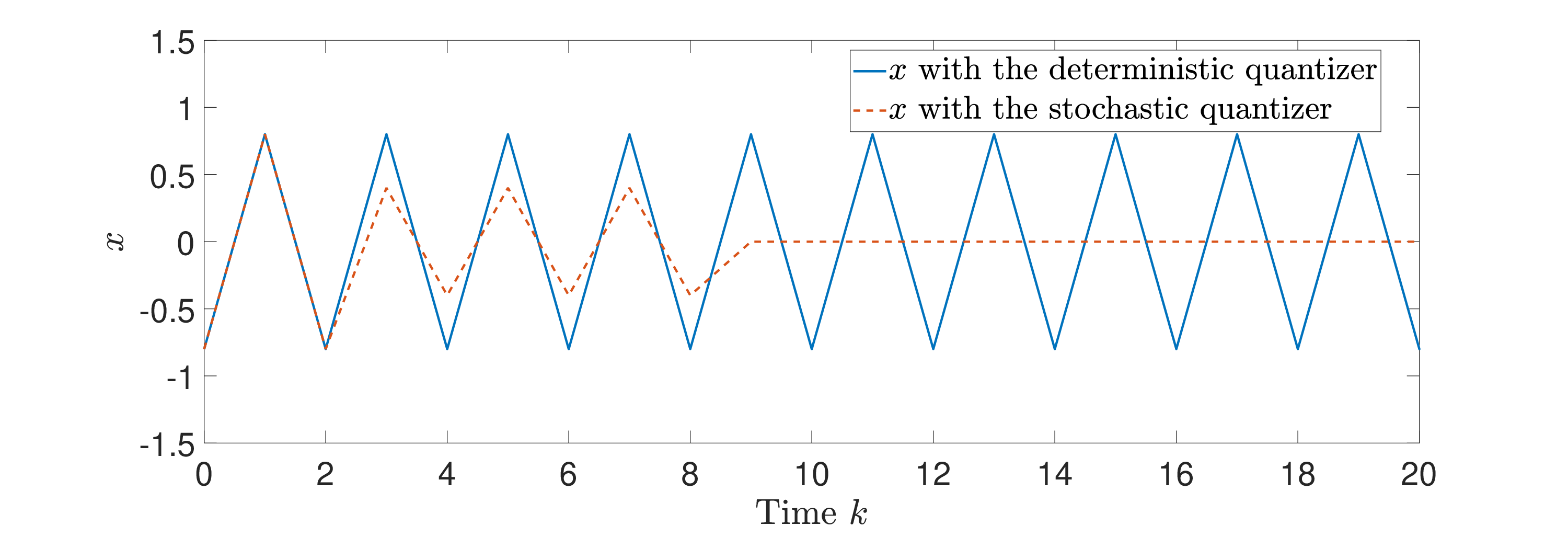}
    \caption{Deterministic quantizer versus stochastic quantizer}
    \label{fig:example_compare}
\end{figure}

\begin{secex} 
Consider the system~\eqref{eq:sys} with
$$
A = -1, \; B = 0.2, \; C = 1, \; H_p = 1.
$$
As the control objective, we consider stabilization. Then, the exo-system~\eqref{reference_sys} is selected by 
$$
A_r = 0, \; H_r = 1. 
$$
The fusion center~\eqref{eq:fusion} is designed by
$$
L = -1, \; K_x = 1, \; K_r = 0.
$$
As a quantizer~\eqref{eq:Q}, we employ the following deterministic uniform quantizer:
\begin{equation}
\label{de_quantizer}
    \mathcal{Q}_v(z+nd) = nd \quad \text{for }  z  \in \left( -\frac{d}{2}, \frac{d}{2} \right], \; n \in \bZ,\; d>0,
\end{equation}
where the quantization step $d$ is selected by $d = 2$. Then, the quantizer does not distinguish any value in $(-1, 1]$. When $x_0 = -0.8$, the quantized output $v(k)$ is zero for all $k \in \bZ_+$. Accordingly, the control input $u(k)$ is identical to zero, and thus, it fails for stabilization as in Fig. \ref{fig:example_compare}. 

Next, we employ the following static stochastic quantizer:
\begin{align}
\label{static}
\left\{\begin{array}{l}
\displaystyle \bP(\mathcal{Q}_v(z+nd) = nd) = 1 - \frac{z}{d},\\[2mm]
\displaystyle \bP(\mathcal{Q}_v(z+nd) = (n+1)d) = \frac{z}{d}
\end{array}\right.\\
\text{for } z \in (0, d], \; n\in \bZ, \; d > 0,
\nonumber
\end{align}
where $d = 2$. As confirmed in Fig. \ref{fig:example_compare}, the stochastic quantizer improves control performance.
\red
\end{secex}

As shown above, a stochastic quantizer can lead to a better control performance by breaking periodic behavior. Moreover, as shown below, introducing stochasity enable us privacy analysis in terms of differential privacy \cite{dwork2006differential} while a privacy concept in the deterministic setting is not well established.
By these two reasons, we in this paper focus on stochastic quantizers to analyze a trade-off between control and privacy performances. We first study the static case and then the dynamic case and show that a dynamic quantizer provides more freedom to balance the trade-off.



\section{Trade-off Analysis for Static Stochastic Quantizers}\label{sec:static}
In this section, we study a trade-off between control and privacy performance when implementing a static stochastic quantizer. As a measure of privacy performance, we employ differential privacy~\cite{dwork2006differential}. As control performance, we evaluate tracking error. Through their lower or upper bounds, we formulate their trade-off.

\subsection{Differential Privacy}
In this subsection, we use differential privacy as a privacy criterion of a system~\eqref{eq:sys} with quantized output $v$ in~\eqref{eq:Q}.
To show its definition, let us define
\begin{align}
O_k &:= 
\begin{bmatrix} C^{\top} & (CA)^{\top} & \cdots & (CA^k)^{\top} \end{bmatrix}^{\top}, \nonumber\\
N_k &:=\left[\begin{array}{ccccc}
0 & 0 & \cdots & \cdots & 0 \\
C B & 0 & \ddots & & \vdots \\
C A B & C B & 0 & \ddots & \vdots \\
\vdots & \vdots & \ddots & \ddots & 0 \\
C A^{k-1} B & C A^{k-2} B & \cdots & C B & 0
\end{array}\right], \label{eq:Nk}\\
U_k &:= \begin{bmatrix} u^{\top}(0) & u^{\top}(1) &\cdots  & u^{\top}(k)\end{bmatrix}^{\top}, \nonumber\\
V_k &:= \begin{bmatrix} v^{\top}(0) & v^{\top}(1) &\cdots  & v^{\top}(k)\end{bmatrix}^{\top}.\nonumber
\end{align}
From~\eqref{eq:sys} and~\eqref{eq:Q}, we have
\begin{align}\label{eq:mech}
V_k = \cQ_v(O_k x_0 + N_k U_k).
\end{align}
We call~\eqref{eq:mech} a \emph{mechanism}, following standard notions in differential privacy~\cite{dwork2006differential,Yu2020}. Another standard notion is the following adjacency relation.

\begin{secdefn} 
Given $\zeta>0$, a pair of initial states $(x_0, x'_0) \in \bR^{n_1} \times \bR^{n_1} $ is said to belong to the \emph{$\zeta$-adjacency relation} if $ | x_0 - x'_0|_{1} \leq \zeta$. The set of all pairs of the initial states that are $\zeta$-adjacent is denoted by $\operatorname{Adj}_1^\zeta$.
\red
\end{secdefn}

A basic idea of differential privacy is to evaluate sensitivity of a mechanism with respect to the variable required to be private. To be more specific, the mechanism is viewed as highly private if for a pair $(x_0, x'_0) \in \operatorname{Adj}_1^\zeta$, the corresponding pair $(V_k, V'_k)$ is sufficiently close; note that $U_k$ is common for $(V_k, V'_k)$ because this can be eavesdropped. Based on this idea, differential privacy of a mechanism induced by a dynamical system is defined as follows by modifying \cite[Definition 2.4]{Yu2020}.

\begin{secdefn} \label{def:privacy}
Let $(\bR^{(k+1) p}, \cF, \bP )$ be a probability space. The mechanism~\eqref{eq:mech} with a stochastic quantizer is said to be \emph{$(\varepsilon, \delta)$-differentially private} for $\operatorname{Adj}_1^\zeta$ at a finite time instant $t \in \bZ_{+}$ if there exist $\varepsilon \geq 0$ and $\delta \geq 0$ such that
\begin{align}
&\bP \left(\mathcal{Q}_v(O_k x_0+N_k U_k) \in \cS\right) \nonumber \\
\label{eq:privacy}
&\leq \mathrm{e}^{\varepsilon} \bP \left(\mathcal{Q}_v(O_k x'_0+N_k U_k) \in \cS\right)+\delta, 
\quad \forall \cS \in \cF
\end{align}
for any $(x_0, x'_0) \in \mathrm{Adj}_1^\zeta$ and $U_k \in \bR^{(k+1)m}$.
\red
\end{secdefn}

In Definition~\ref{def:privacy}, a distance between $(V_k, V'_k)$ is evaluated in the stochastic sense. If $\varepsilon , \delta \ge 0$ are small, the distance is close, i.e., a mechanism~\eqref{eq:mech} is highly private.


\subsection{Privacy Analysis}
As a stochastic quantizer, we use the one in~\eqref{static}. It is expected that the initial state tends to be more private when the quantization step $d$ is larger, which is shown to be true. Moreover, we provide a lower bound on the quantization step $d$ to achieve a desired differential privacy level.

For the linear system $x(k+1) = A x(k)$ with $A$ in~\eqref{eq:sys}, it is well known that there exist $\beta$, $\lambda > 0$ such that
\begin{align}\label{eq:sys_norm}
|x(k) - x'(k)|_1 \leq \beta \lambda^k |x_0 - x'_0|_1, \quad \forall k \in \bZ_+
\end{align}
for all $(x_0, x'_0) \in \bR^{n_1} \times \bR^{n_1}$. Note that $A$ is not necessarily to be Schur stable, i.e., $\lambda$ can be greater than or equal to $1$. 

Using $\beta, \lambda > 0$, we can derive a lower bound on the quantization step to achieve a desired differential privacy level. This estimation stated below is the first main result of this paper.

\begin{secthm}\label{thm:dp}
Given $k \in \bZ_{+}$, $\zeta > 0$, and $\delta \in (0, 1)$, a mechanism~\eqref{eq:mech} with a static stochastic quantizer~\eqref{static} is $(0,\delta)$ differentially private for $\operatorname{Adj}_1^\zeta$ at $k$ if the quantization step $d$ satisfies
\begin{align}\label{eq1:dp}
d \geq \sum_{t = 0}^{k} \frac{\beta |C|_1 \lambda^{t} \zeta}{\delta}.
\end{align}
\end{secthm}

\begin{proof}
The proof is in Appendix~\ref{app1}.
\end{proof}

The above theorem shows that for the considered mechanism, we obtain $(0,\delta)$ differential privacy, i.e., $\varepsilon = 0$. For making the quantization step $d$ large, we can achieve differential privacy for arbitrary $\delta > 0$. If $\beta$, $\lambda > 0$ in \eqref{eq:sys_norm} are small, the quantization step can be selected smaller to achieve the same differential privacy level. Moreover, if $A$ is Schur stable, i.e., $\lambda \in (0,1)$, we can take $k \to \infty$ as stated below.

\begin{seccor} \label{cor:dp}
    Given $\zeta > 0$ and $\delta \in (0, 1)$, a mechanism~\eqref{eq:mech} with a static stochastic quantizer~\eqref{static} is $(0,\delta)$ differentially private for $\operatorname{Adj}_1^\zeta$ for any $k \in \bZ_+$ if the quantization step $d$ satisfies
\begin{align}\label{eq2:dp}
d \geq \frac{\beta |C|_1 \zeta}{(1- \lambda)\delta}.
\end{align}
\end{seccor}

\begin{proof}
 The statement can be shown by taking $k \to \infty$ in~\eqref{eq1:dp} with $\lambda \in (0, 1)$.
\end{proof}


\subsection{Trade-off between Control and Privacy Performances} 
    From Theorems \ref{thm:dp} and \ref{cor:dp}, a mechanism~\eqref{eq:mech} with a static stochastic quantizer~\eqref{static} can be made more private by selecting the quantization step $d$ larger. However, this degenerates control performance. In this subsection, we investigate a trade-off between control and privacy performances.
    
   As a criterion of control performance, we evaluate the tracking error $e_y  = y_p - y_r$ in~\eqref{eq:conob}, caused by quantization:
\begin{equation}\label{J}
    J =  \lim_{k \to \infty}  \bE [e_y ^{\top}(k) Q e_y (k)]
\end{equation}
with a weight $0 \preceq Q \in \bR^{n \times n}$.

Due to the nonlinearity of a quantizer \eqref{static}, it is difficult to compute the exact value of $J$ in \eqref{J}. However, as one can observe from the form of a stochastic quantizer~\eqref{static}, the error due to the quantization is bounded. This enables us to calculate the upper bound on $J$ as follows. 

\begin{secthm}\label{thm:J}
Under Assumptions~\ref{asm:1}-\ref{asm:4}, consider the closed-loop system consisting of a system~\eqref{eq:sys} and fusion center~\eqref{eq:fusion} with a static stochastic quantizer~\eqref{static}. Let $K_x$ and $L$ be designed such that $A + B K_x$ and $A + L C$ are Schur stable, respectively. Also, select $K_r := U - K_x X$ for $X$ and $U$ in Assumption~\ref{asm:4}. Then, for any given $Q \succeq 0$, the control performance index $J$ in~\eqref{J} is upper bounded as in
\begin{align}\label{eq:Jub}
J \leq \frac{d^2}{2} {\rm trace}(H_p^\top Q  H_p) {\rm trace} (Z)
\end{align}
for any initial state $(x(0), x_r(0), \hat x(0)) \in \bR^{n_1} \times \bR^{n_2} \times \bR^{n_1}$, where $Z \succeq 0$ is a solution to
\begin{align}
&Z = \cA  Z \cA ^{\top}  + \begin{bmatrix}
        I \\ I
    \end{bmatrix}
    L L^\top
    \begin{bmatrix}
        I \\ I
    \end{bmatrix}^\top, \label{eq:Lya}\\
&\quad
\cA  := \begin{bmatrix}
        A+BK_x & LC \\
        0 & A+LC
    \end{bmatrix}.\label{eq:cA}
\end{align}
\end{secthm}
\begin{proof}
The proof is in Appendix~\ref{app2}.
\end{proof}

The upper bound on $J$ in~\eqref{eq:Jub}  is an increasing function of $d$. In contract, the lower bound \eqref{eq1:dp} on the differential privacy level is an decreasing function of $d$. To make this trade-off clearer, we now substitute the smallest $d$ satisfying \eqref{eq1:dp} into \eqref{eq:Jub}, yielding 
\begin{align*}
J \leq \frac{1}{2} \left( \sum_{t = 0}^{k} \frac{\beta |C|_1 \lambda^{t} \zeta}{\delta} \right)^2 {\rm trace}(H_p^\top Q  H_p) {\rm trace} (Z).
\end{align*}
This shows that increasing privacy performance, i.e., decreasing $\delta$ leads to degeneration of the control performance.



\section{Improving Trade-off by Dynamic Stochastic Quantizers}\label{sec:dsq}
As discussed in the previous section, there is a trade-off between control and privacy performances. In this section, we show that this trade-off can be improved by employing a dynamic stochastic quantizer when $A$ is Schur stable.

At the beginning of control, the quantization step can be made large to protect the initial state from being inferred and then can be made small to improve the control performance. Based on this idea, we focus on designing a \emph{zoom-in} dynamic stochastic quantizer given by 
\begin{align}
\label{dynamic}
\left\{\begin{array}{l}
\displaystyle \bP(\mathcal{Q}_v(z+nd(k)) = nd(k)) = 1 - \frac{z}{d(k)},\\[2mm]
\displaystyle \bP(\mathcal{Q}_v(z+nd(k))) = (n+1)d(k)) = \frac{z}{d(k)}
\end{array}\right.\\
\text{for } z \in (0, d(k)], \; n\in \bZ, \; d(k) > 0,
\nonumber
\end{align}
where for given $d^* \in [0, d(0)]$,
\begin{align*}
\lim_{k \rightarrow \infty} d(k) = d^*.
\end{align*}
Also, $d(k)$ is a decreasing function with respect to $k$ under the following constraint determined by $0 < q < 1$:
\begin{align}
\label{constraint_quantizer}
    (d(0) - d^*) q^k \le d(k)-d^* \le d(0)- d^*.
\end{align}
This constraint means that the quantization step cannot decrease too fast. Such a quantizer is designed in, e.g., \cite[Theorem 9]{Brockett_quantizer}.

As the main result of this section, we extend Theorem~\ref{thm:dp} for differential privacy analysis to the dynamic stochastic quatnizer as follows.
\begin{secthm}
\label{thm:dyn_quantizer_arbi}
     Given $k \in \bZ_{+}$, $\zeta > 0$, and $\delta \in (0, 1)$, a mechanism~\eqref{eq:mech} with a dynamic stochastic quantizer~\eqref{dynamic} under the constraint \eqref{constraint_quantizer} is $(0,\delta)$ differentially private for $\operatorname{Adj}_1^\zeta$ at $k$ if one of the following two conditions holds:
     \begin{enumerate}
     \item the quantization step $d(0)$ satisfies 
     \begin{align}\label{eq1:dp_dyn}
d(0) \geq \sum_{t = 0}^{k} \frac{\beta |C|_1 \lambda^{t} \zeta}{q^t \delta},
\end{align}
     and $\lambda < q < 1 $, where $\beta, \lambda > 0$ are of~\eqref{eq:sys_norm};
     
     \item $d^*$ satisfies 
     \begin{align}
         d^* \geq \sum_{t = 0}^{k} \frac{\beta |C|_1 \lambda^{t} \zeta}{\delta}.
     \end{align}
     \end{enumerate}
\end{secthm}

\begin{proof}
The proof is in Appendix~\ref{app3}.
\end{proof}

Similarly, we allow $k$ to infinity when the system matrix $A$ is Schur stable, which is stated as below.

\begin{seccor}
\label{thm:dyn_quantizer}
     Given $\zeta > 0$ and $\delta \in (0, 1)$, a mechanism~\eqref{eq:mech} with a dynamic stochastic quantizer~\eqref{dynamic} under the constraint \eqref{constraint_quantizer} is $(0,\delta)$ differentially private for $\operatorname{Adj}_1^\zeta$ for any $k \in \bZ_+$ if one of the following two conditions holds:
     \begin{enumerate}
     \item quantization step $d(0)$ satisfies 
     \begin{align}\label{eq1:dp_dyn2}
d(0) \geq  \frac{\beta |C|_1 q \zeta}{(q-\lambda) \delta},
\end{align}
     and $\lambda < q < 1 $, where $\beta, \lambda > 0$ are of~\eqref{eq:sys_norm};
     
     \item $d^*$ satisfies
     \begin{align}\label{eq2:dp_dyn2}
d^{*} \geq \frac{\beta |C|_1 \zeta}{(1- \lambda)\delta}.
\end{align} 
\end{enumerate}
\end{seccor}
\begin{proof}
The statement can be shown by taking $k$ to $\infty$ in~\eqref{eq1:dp_dyn} with $\lambda < q$.
\end{proof}

\begin{secrem}
The static case can be recovered, i.e., Theorem~\ref{thm:dyn_quantizer_arbi} (resp. Corollary~\ref{thm:dyn_quantizer}) reduces to Theorem~\ref{thm:dp} (Corollary~\ref{cor:dp}) by selecting $q=1$. 
\red
\end{secrem}

When $A$ is Schur stable, i.e., $\lambda < 1$, the mechanism can be made differential private by selecting $d(0)$ large by item 1). On the other hand, similarly to Theorem~\ref{thm:J}, an upper bound on the control performance $J$ in~\eqref{J} is estimated by
\begin{align}\label{eq:Jd*}
J \leq \frac{(d^*)^2}{2} {\rm trace}(H_p^\top Q  H_p) {\rm trace} (Z).
\end{align}
Consequently, the control performance $J$ can be improved by selecting $d^*$ small. In summary, when $A$ is Schur stable, both privacy and control performances can be improved by utilizing the dynamic stochastic quatnizer with large $d(0)$ and small $d^*$. This suggests to select a large initial quantization step for increasing the privacy performance and to decrease the quantization step over time for improving the control performance.



\section{Privacy Protection for Unstable Systems}\label{sec:ust}
As shown in the above sections, a static/dynamic stochastic quantizer can achieve the prescribed differential privacy level for any time instant only when $A$ is Schur stable. In this section, we discuss the case where $A$ is unstable by applying input noise in addition to a stochastic quantizer. Then, we again show that a dynamic quantizer gives a better trade-off than a static quantizer.

Let $w(k) \sim \mathcal{N}(0, \sigma^2(k) I_m)$ be independent Gaussian noise, where the variance $\sigma^2(k)$ is a tuning parameter. We add $w(k)$ to the input channel of~\eqref{system_1}. Namely, we consider the following system model instead of~\eqref{system_1}: 
\begin{align}
\label{Gaussian_mechanism}
     x(k+1) = Ax(k) + B(u(k) + w(k)),
\end{align}
where the input $u$ is designed as in the previous sections. From~\eqref{eq:mech}, the new mechanism is
\begin{align}\label{eq:mech2}
&V_k = \cQ_v(O_k x_0 + N_k (W_k + U_k)),\\
&\quad W_k := \begin{bmatrix} w^{\top}(0) & w^{\top}(1) &\cdots  & w^{\top}(k)\end{bmatrix}^{\top}.
\nonumber
\end{align}

In fact, applying input Gaussian noise first several steps in addition to the use of a stochastic quantizer ensures the initial state privacy under a desired differential privacy level, stated below.

\begin{secthm}
\label{thm:input_noise}
    Suppose that a pair $(A, B)$ in~\eqref{Gaussian_mechanism} is controllable. Define $M := [A^{n^*-1}B, \cdots, AB, B]$ and $\Delta := MM^{\top}$, where $n^*$ is the smallest integer such that $\Delta$ is non-singular. Furthermore, assume $CA^{k}B = 0$ for $0\leq k \leq n^{*}-2$. Then, a mechanism~\eqref{eq:mech2} with a dynamic stochastic quantizer~\eqref{dynamic} is $(\varepsilon, \delta)$ differentially private for $\operatorname{Adj}_1^\zeta$ for any $k \in \bZ_+$ if the parameters are designed as follows:
    \begin{enumerate}
    \item $\varepsilon \ge \varepsilon_0$ and $\delta \ge \delta_1+\delta_2$;
        \item $d(k)$ satisfies~\eqref{constraint_quantizer} and
        \begin{align*}
        d(0) \geq \sum_{t = 0}^{n^{*} - 1}\frac{\beta |C|_1 \lambda^{t}\zeta}{\delta_1 q^{t}};
        \end{align*}
        \item  $\sigma(k)$ satisfies
         \begin{align*}
        \sigma(k) = \begin{cases} \displaystyle \sigma \geq \frac{| \Delta^{-\frac{1}{2}}A^{n^*}|_2 \zeta}{\kappa_{\varepsilon_0}^{-1}(\delta_2)}, & 0 \leq k < n^{*} \\ 0, & k \ge n^{*} 
        \end{cases}
        \end{align*}
    \end{enumerate}
    where $\kappa_{\varepsilon_0}^{-1}(\delta_2)$ is the inverse function of 
\begin{align*}
&\kappa(\varepsilon_0, \delta_2):=\alpha \left(\frac{\delta_2}{2}-\frac{\varepsilon_0}{\delta_2}\right)-e^{\varepsilon_0} \alpha \left(-\frac{\delta_2}{2}-\frac{\varepsilon_0}{\delta_2}\right)\\
&\quad \alpha(a):=\frac{1}{\sqrt{2 \pi}} \int_{-\infty}^a e^{-\frac{w^2}{2}} d w
\end{align*}
with respect to $\delta_2$.
\end{secthm}
\begin{proof}
The proof is in Appendix~\ref{app4}.
\end{proof}

\begin{secrem}\label{rem:satic}
Theorem~\ref{thm:input_noise} can be applied to a static stochastic quantizer~\eqref{static} by selecting $q = 1$.
\red
\end{secrem}

Since Gaussian noise is added up to $n^*-1$ time instant, the control performance~\eqref{J} is determined by a quntization step. Namely, its upper bound is estimated as~\eqref{eq:Jd*} by using the terminal quanitization step $d^*$ that can be made arbitrary small. Thus, again selecting $d(0)$ large and $d^*$ small improves a trade-off between privacy and control performances, where they are the same for a static stochastic quantizer. This implies that a dynamic quantizer gives a better trade-off performance.





\section{Simulations}\label{sec:sim}
In this section, we illustrate proposed privacy-preserving methods through a motion control problem, where the local system is an autonomous car that receives control inputs from a service provider. The initial state can represent the car owner's home address and the initial speed; the owner's home address is private information. 

The system parameters in~\eqref{eq:sys} and \eqref{reference_sys} are given as follows:
\begin{align*}
    &A = \begin{bmatrix}
    1 & 0 & \tau & 0\\ 0 & 1 & 0 & \tau \\ 0& 0 & 0 & 0 \\ 0& 0 & 0 & 0
\end{bmatrix}, \;
B = \begin{bmatrix}
    0 & 0 \\ 0 & 0 \\ 1 & 0 \\ 0 & 1
\end{bmatrix}, \;
C = \begin{bmatrix}
    1 & 0 & 0 & 0 \\ 0 & 1 & 0 & 0
\end{bmatrix},  \\
&\tau = 0.1, \; H_p = C, \; A_r = I_2, \; H_r = I_2,\\
&x_r(t)  \equiv \begin{bmatrix}
    10 & 10
\end{bmatrix}^{\top}.
\end{align*}
The fusion center is designed as follows:
\begin{align*}
K_x &= \begin{bmatrix}
    -1 & 0 & -1 & 0 \\
    0 & -1 & 0 & -1
\end{bmatrix}, \;
K_r = \begin{bmatrix}
    1 & 0 \\ 0 & 1
\end{bmatrix},\\
L &= 
\begin{bmatrix}
    -0.7238 & 0 & -0.0020 & 0 \\ 0 & -0.7238 & 0 & -0.0020
\end{bmatrix}^{\top}.
\end{align*}
To estimate the upper bound on the control performance index $J$ in~\eqref{eq:Jub} for $Q=I_2$, we compute
\begin{align*}
{\rm trace}(Z) = 3.3135, \quad {\rm trace}(H^{\top}_p Q H_p) = 2.
\end{align*}

We apply Theorem \ref{thm:input_noise} to design a privacy-preserving mechanism.
The assumptions in Theorem \ref{thm:input_noise} hold because a pair $(A, B)$ is controllable, $n^* = 2$, and $CB =0$.
First, we design a static stochastic quantizer~\eqref{static} by selecting $q=1$; recall Remark~\ref{rem:satic}. One can calculate $\beta =1$ and $\lambda = 1$ in~\eqref{eq:sys_norm}. Then, $d(0) = 4$ and $\sigma^2 = 5$ satisfy the conditions in Theorem \ref{thm:input_noise}, where $\varepsilon_0 = 0.3$, $\delta_1 = 0.05$, and $\delta_2=0.0461$. Namely, the mechanism~\eqref{eq:mech2} is $(\varepsilon, \delta)$ differentially private for $\varepsilon = 0.3$ and $\delta = 0.0961$. Also, the upper bound on the control performance is $J \leq 53.0$.

Next, we design a dynamic stochastic quantizer~\eqref{dynamic} by selecting $d(0) = 10$, $d^* =0$, and $q= 0.99$. Then, $\delta_1 = 0.0199$, and thus the mechanism~\eqref{eq:mech2} is $(\varepsilon, \delta)$ differentially private for $\varepsilon = 0.3$ and $\delta = 0.0660$. Thus, the dynamic quantizer provides a more private mechanism than the static one. Also, the control performance is $J = 0$, which is smaller than that in static one. This shows that the dynamic quantizer has better performances for both control and privacy.


Figures~\ref{fig:sto_x_1} and~\ref{fig:sto_x_2} show the tracking outputs for $x_0 = 0$ when static and dynamic stochastic quantizers are implemented. It is observed that the dynamic quantizer has better control performances while making the mechanism more private than the static one. Therefore, the dynamic quantizer improves both privacy and control performances.

\begin{figure}[t]
    \centering    \includegraphics[width = 0.5\textwidth]{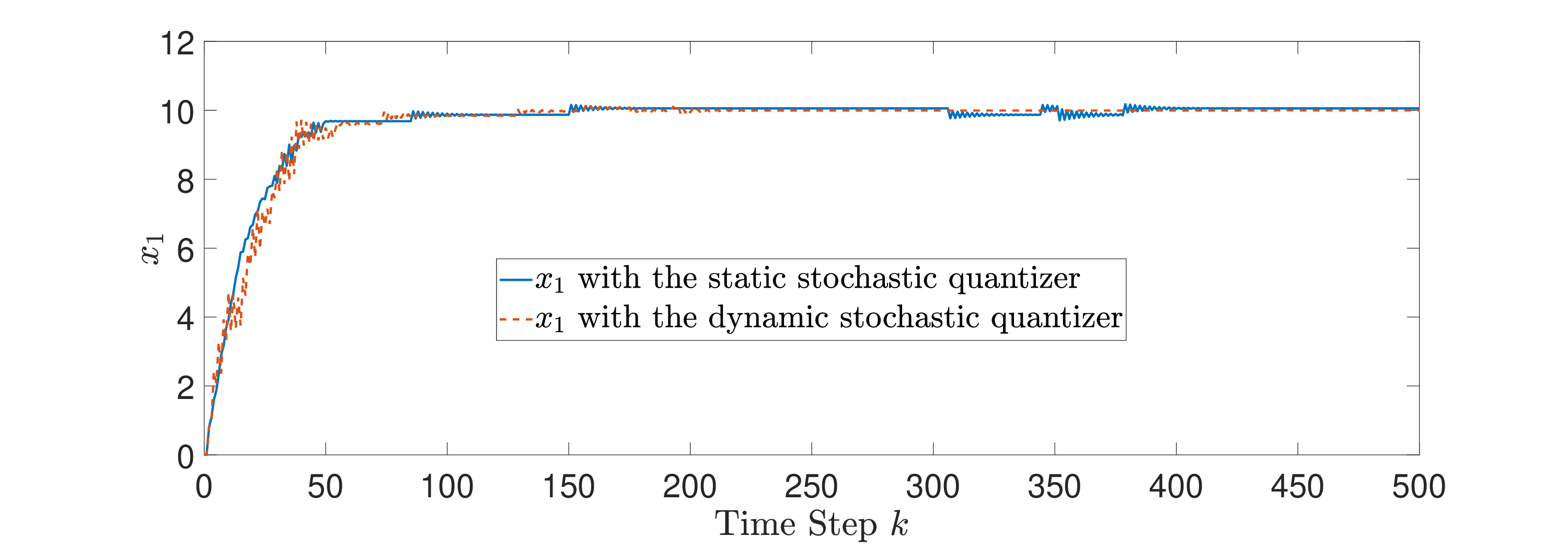}
    \caption{Dynamic versus static stochastic quantizers: tracking output $y_1 = x_1$ }
    \label{fig:sto_x_1}

    \centering    \includegraphics[width = 0.5\textwidth]{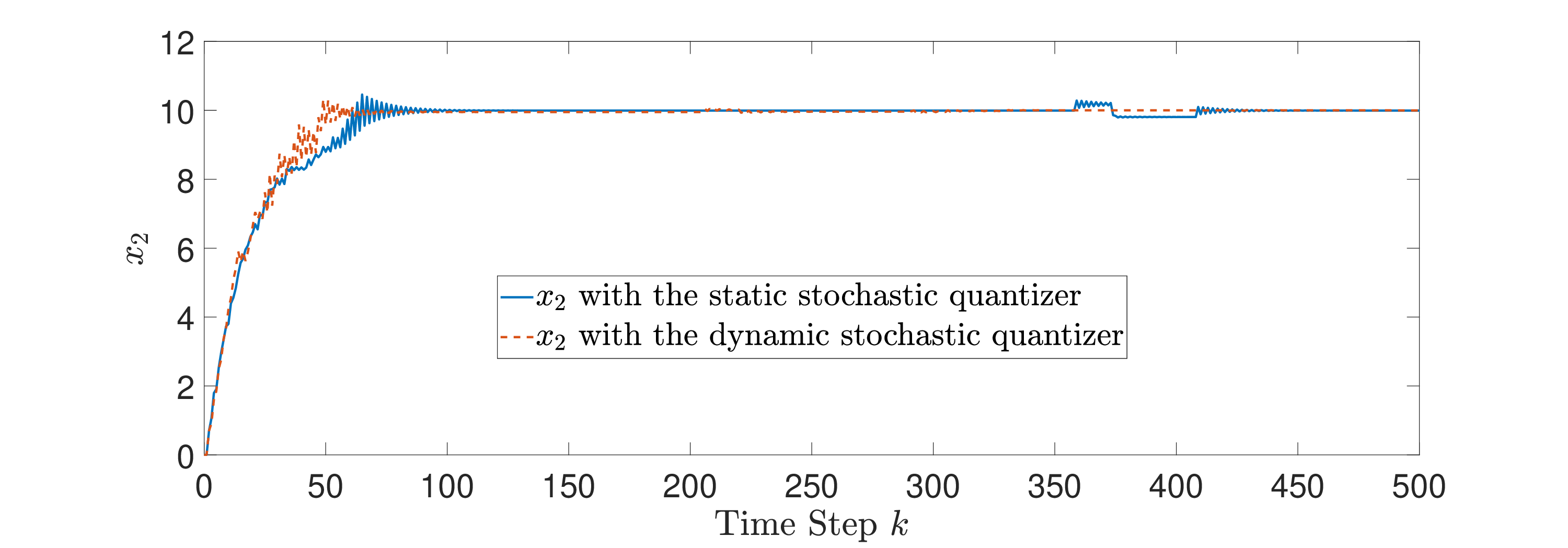}
    \caption{Dynamic versus static stochastic quantizers: tracking output $y_2 = x_2$ }
    \label{fig:sto_x_2}
\end{figure}



\section{Conclusions}\label{sec:con}
In this paper, we have studied the role of stochastic quantizers for privacy protection in the context of networked tracking control. First, focusing on a static stochastic quantizer, we have derived a sufficient condition for a quantization step to achieve $(0,\delta)$ differential privacy. Then, we have estimated an upper bound on the tracking error caused by the quantization. These privacy and control performance are in the relation of a trade-off, determined by the quantization step. To improve this trade-off, we have introduced a dynamic stochastic quantizer and illustrated the improvement by selecting the initial and terminal quantization steps appropriately when a system matrix is Schur stable. To deal with the unstable case, we have further provided a new mechanism by the combination of input Gaussian noise and stochastic quantizers and shown that a dynamic quantizer again achieves a better trade-off performance.



\appendices

\section{Proof of Theorem~\ref{thm:dp}}\label{app1}

Before proving the theorem, we first provide the following auxiliary lemma.

\begin{seclem}\label{lem1}
Consider a stochastic quantizer~\eqref{eq:Q} and~\eqref{static} with scalar $y$. If $s \leq d$, it follows that
\begin{align*}
\sup _{|y-y'| \leq \ s}\left| \bP \left(v \in \cS \mid y\right) - \bP \left(v'\in \cS \mid y'\right)  \right| \leq \frac{s}{d}, \quad \forall S \in \cF
\end{align*}
for any $(y, y') \in \bR \times \bR$.
\end{seclem}

\begin{proof}
It suffices to consider the following two cases
\begin{itemize}
    \item [] Case 1: $y, y' \in (nd, (n+1)d]$;
    \item [] Case 2: $y \in (nd, (n+1)d]$, $y' \in ((n+1)d, (n+2)d]$.
\end{itemize}

(Case 1) 
In this case, we have
\begin{align*}
\cF = \{\emptyset, \{nd\}, \{(n+1)d\}, \{nd, (n+1)d\}\}.
\end{align*}
First, for $S = \emptyset$ and $S = \{nd, (n+1)d\}$, we obtain
\begin{align*}
\bP \left(v \in \cS \mid y\right) - \bP \left(v'\in \cS \mid y'\right) = 0.
\end{align*}

Next, for $\cS = \{nd\}$, it follows that
\begin{align*}
&\sup _{|y-y'| \leq s}\left| \bP \left(v = nd \mid y\right) - \bP \left(v'=nd \mid y'\right)  \right| \\
&= \sup _{|y-y'| \leq s}\left|  1-\frac{y-nd}{d} - \left( 1-\frac{y'-nd}{d} \right) \right| \\
&= \sup _{|y-y'| \leq s}\left| \frac{y - y'}{d} \right|  
\leq  \frac{s}{d}.
\end{align*}

Finally, for $\cS = \{(n+1)d\}$, it holds that
\begin{align*}
&\sup _{|y-y'| \leq s}\bigg| \bP \left(v=(n+1)d \mid y\right) - \bP \left(v'=(n+1)d \mid y'\right)  \bigg| \\
& = \sup _{|y-y'| \leq s}\left|  \frac{y-nd}{d} - \frac{y'-nd}{d}\right|
 = \sup _{|y-y'| \leq s}\left| \frac{y - y'}{d} \right| 
\leq  \frac{s}{d}.
\end{align*}

(Case 2)
In this case, we have
\begin{align*}
\cF &= \{\emptyset, \{nd\}, \{(n+1)d\}, \{(n+2)d\},  \\
&\qquad \{nd, (n+1)d\}, \{nd, (n+2)d\},  \\
&\qquad \{(n+1)d, (n+2)d\}, \{nd, (n+1)d, (n+2)d\}\}.
\end{align*}
First, for $S = \emptyset$ and $S = \{nd, (n+1)d, (n+2)d\}$, we obtain
\begin{align*}
\bP \left(v \in \cS \mid y\right) - \bP \left(v'\in \cS \mid y'\right) = 0.
\end{align*}

Second, for $\cS = \{nd\}$, it follows that
\begin{align*}
&\sup _{|y-y'| \leq s}\bigg| \bP \left(v=nd \mid y\right) - \bP \left(v'=nd \mid y'\right)  \bigg| \\
&= \sup _{|y-y'| \leq s}\left|  1 - \frac{y-nd}{d} \right| 
= \sup _{|y-y'| \leq s}\left| \frac{y-(n+1)d}{d} \right|.
\end{align*}
From $y \in (nd, (n+1)d]$, we have 
\begin{align*}
0 \ge y - (n+1)d. 
\end{align*}
From $y' \in ((n+1)d, (n+2)d]$ when $0 < y' - y \le s$, we obtain
\begin{align*}
0 \ge y - (n+1)d \ge  - s.
\end{align*}
Thus, it follows when $y' - y \le s$ that
\begin{align*}
\sup _{|y-y'| \leq s}\left| \frac{y-(n+1)d}{d} \right| \le \frac{s}{d}.
\end{align*}
The case where $\cS = \{(n+2)d\}$ can be shown similarly.

Third, for $\cS = \{(n+1)d\}$, we have
\begin{align*}
&\sup _{|y - y'| \leq s}\bigg| \bP \left(v=(n+1)d \mid y\right)  - \bP \left(v'=(n+1)d \mid y'\right)  \bigg| \\
&= \sup _{|y-y'| \leq s}\left|  \frac{y-nd}{d} - \left(1- \frac{y'-(n+1)d}{d}\right)\right| \\
&= \sup _{|y-y'| \leq s}\left| \frac{y + y' - (2n+2)d }{d} \right|.
\end{align*}
Since $y \in (nd, (n+1)d]$ and $|y-y'| \leq s$, it follows that 
\begin{align*}
    y+y^{\prime} \leq 2y+ |y-y^{\prime}|\leq 2(n+2)d + s.
\end{align*}
Also, since $y' \in ((n+1)d, (n+2)d]$ , we obtain
\begin{align*}
    (2n+2)d - s < 2y^{\prime} - |y -y^{\prime}| \leq y + y'
\end{align*}
Therefore, we have
\begin{align*}
 \sup _{|y-y'| \leq s}\left| \frac{y + y' - (2n+2)d }{d} \right| \leq \frac{s}{d}.
\end{align*}

Fourth, for $\cS = \{nd, (n+1)d\}$, we have
\begin{align*}
    & \sup _{|y-y'| \leq s}\bigg| \bP \left(v \in \cS \mid y\right) - \bP \left(v'\in \cS \mid y'\right)  \bigg| \\
    &=  \sup _{|y-y'| \leq s}\bigg| 1 - \bP \left(v'=(n+1)d \mid y'\right)  \bigg| \leq \frac{s}{d}.
\end{align*}
The case where $\cS = \{(n+1)d, (n+2)d\}$ can be shown similarly.

Finally, for $\cS = \{nd, (n+2)d\}$, we have 
\begin{align*}
&\sup _{|y-y'| \leq s}\bigg| \bP \left(v \in \cS \mid y\right) - \bP \left(v' \in \cS \mid y'\right)  \bigg| \\
&=\sup _{|y-y'| \leq s}\bigg| \bP \left(v=nd \mid y\right) - \bP \left(v'=(n+2)d \mid y'\right)  \bigg| \\
&= \sup _{|y-y'| \leq s}\left|  1 - \frac{y-nd}{d} - \frac{y'-(n+1)d}{d} \right| \\
&= \sup _{|y-y'| \leq s}\left| \frac{y+y'-2(n+1)d}{d} \right|
\le  \frac{s}{d}.
\end{align*}
That completes the proof.
\end{proof}


Now, we are ready to prove Theorem~\ref{thm:dp}.

{\it Proof of Theorem~\ref{thm:dp}:}
Note that $\cS$ can be described as $\cS = \cS_0 \times \cdots \times \cS_{k}$, where $\cS_t$ denotes the possible output at $t$-th time instant. Since the output of the stochastic quantizer~\eqref{static} is i.i.d., we have
\begin{align}\label{pf1:dp}
\bP \left(V_k \in \cS \mid x_0, U_k\right) = \prod_{t = 0}^{k} \bP \left(v(t)\in \cS_t \mid x_0, U_k\right).
\end{align}
Since $U_k$ is public information, and $x_0$ is needed to be private, we rewrite it, for the sake of rotational simplicity, by
\begin{align*}
\bP \left(v(t)\in \cS_t \mid x_0\right) := \bP \left(v(t)\in \cS_t \mid x_0, U_k\right).
\end{align*}
Also, for the sake of rotational simplicity, define 
\begin{align}\label{pf1.5:dp}
s := \sum_{t=0}^{k} \beta |C|_1 \lambda^{t} \zeta.
\end{align}

It follows from~\eqref{pf1:dp} that
\begin{align}\label{pf2:dp}
&\sup _{|x_0-x'_0|_1 \leq \zeta }\left|\bP \left(V_k \in \cS \mid x_0\right) - \bP \left(V'_k\in \cS \mid x'_0\right)\right| \nonumber\\
&\leq \sup _{|y_k-y'_k|_1 \leq s }\left|\bP \left(V_k \in \cS \mid x_0\right)-\bP \left(V'_k\in \cS \mid x'_0\right)\right| \nonumber\\
&=  \sup _{|y_k-y'_k|_1 \leq s } \Biggl| \prod_{t = 0}^{k} \bP \left(v(t) \in \cS_t \mid x_0\right) \nonumber\\
&\hspace{30mm} - \prod_{t = 0}^{k} \bP \left(v'(t) \in \cS_t \mid x'_0\right) \Biggr|
\end{align}
The most right hand side can be rearranged as
\begin{align*}
&\prod_{t = 0}^{k} \bP \left(v(t) \in \cS_t \mid x_0\right)  - \prod_{t = 0}^{k} \bP \left(v'(t) \in \cS_t \mid x'_0\right)\\
&= \bP \left(v(0) \in \cS_0 \mid x_0\right) \prod_{t = 1}^{k} \bP \left(v(t)\in \cS_t \mid x_0\right)\\
&\qquad -  \bP \left(v'(0)\in \cS_0 \mid x'_0\right) \prod_{t = 1}^{k} \bP \left(v'(t)\in \cS_t \mid x'_0\right)\\
&= ( \bP \left(v(0) \in \cS_0 \mid x_0\right) - \bP \left(v'(0)\in \cS_0 \mid x'_0\right) )\\
&\qquad \times\prod_{t = 1}^{k} \bP \left(v(t) \in \cS_t \mid x_0\right) \\
&\quad + \bP \left(v'(0)\in \cS_0 \mid x'_0\right) \\
&\qquad   \times \left(\prod_{t = 1}^{k} \bP \left(v(t)\in \cS_t \mid x_0\right) - \prod_{t =1}^{k} \bP \left(v'(t)\in \cS_t \mid x'_0\right) \right).
\end{align*}
Taking the absolute value leads to
\begin{align*}
&\left| \prod_{t = 0}^{k} \bP \left(v(t) \in \cS_t \mid x_0\right)  - \prod_{t = 0}^{k} \bP \left(v'(t) \in \cS_t \mid x'_0\right) \right|\\
&\le \left| \bP \left(v(0) \in \cS_0 \mid x_0\right) - \bP \left(v'(0)\in \cS_0 \mid x'_0\right) \right|\\
&\qquad \times \left| \prod_{t = 1}^{k} \bP \left(v(t) \in \cS_t \mid x_0\right) \right| \\
&\quad + \left| \bP \left(v'(0)\in \cS_0 \mid x'_0\right) \right| \\
&\qquad   \times \left|\prod_{t = 1}^{k} \bP \left(v(t)\in \cS_t \mid x_0\right) - \prod_{t =1}^{k} \bP \left(v'(t)\in \cS_t \mid x'_0\right) \right|\\
&\le \left| \bP \left(v(0) \in \cS_0 \mid x_0\right) - \bP \left(v'(0)\in \cS_0 \mid x'_0\right) \right|\\
&\quad + \left|\prod_{t = 1}^{k} \bP \left(v(t)\in \cS_t \mid x_0\right) - \prod_{t =1}^{k} \bP \left(v'(t)\in \cS_t \mid x'_0\right) \right|,
\end{align*}
where in the second inequality, $|\bP (\cdot)| \le 1$ is used. Repeating this for $t=1,\dots, k$ yields
\begin{align}\label{pf3:dp}
&\left| \prod_{t = 0}^{k} \bP \left(v(t) \in \cS_t \mid x_0\right)  - \prod_{t = 0}^{k} \bP \left(v'(t) \in \cS_t \mid x'_0\right) \right| \nonumber\\
&\le \sum_{t = 0}^{k} \left| \bP \left(v(t)\in \cS_t \mid x_0\right) - \bP \left(v'(t)\in \cS_t \mid x'_0\right)  \right|.
\end{align}

On the other hand, it follows that
\begin{align}\label{pf4:dp}
    & | \bP \left(v(t)\in \cS_t \mid x_0\right) - \bP \left(v'(t)\in \cS_t \mid x'_0\right)| \nonumber\\
 & \leq \sum_{i = 1}^{p} | \bP \left(v_i(t)\in \cS_{i,t} \mid x_0\right) - \bP \left(v'_i(t)\in \cS_{i,t} \mid x'_0\right) \nonumber\\
 & \leq \frac{|y(t) - y'(t)|_1}{d}
\end{align}
for any $(y(t), y'(t)) \in \bR^p \times \bR^p$ such that $|y(t) - y'(t)|_1 \le \beta |C|_1 \lambda^t \zeta < d$,
where $v_i(t)$ is the $i$-th component of $v(t)$. The last inequality follows from Lemma~\ref{lem1}.

Combining \eqref{pf1.5:dp} -- \eqref{pf4:dp} leads to
\begin{align*}
&\sup_{|x_k-x'_k|_1 \leq \zeta }\left|\bP \left(v_k\in \cS \mid x_0\right) - \bP \left(V'_k \in \cS \mid x'_0\right)\right| \\
&\leq \sup _{|y_k-y'_k|_1 \leq s } \sum_{t = 0}^{k} \frac{|y(t) - y'(t)|_1}{d} 
\le \sum_{t = 0}^{k} \frac{\beta |C|_1 \lambda^{t} \zeta}{d}
 \leq \delta.
\end{align*}
The last inequality follows from \eqref{eq1:dp}.
\red



\section{Proof of Theorem~\ref{thm:J}}\label{app2}

First, we have the following lemma for the quantization error $w_v := \mathcal{Q}_v(y) - y$. 

\begin{seclem} \label{lem:2.1}
For a stochastic quantizer $Q_v$ defined in \eqref{static},  the following three holds at each $k \in \bZ_+$:
\begin{enumerate}
\renewcommand{\labelenumi}{(\roman{enumi})}
\item $\bE [w_v(k)] = 0$;
\item $\bE [w_v(k)w_v^{\top}(k)] \preceq ( d^2/4)I$;
\item $\bE [w_v(k)w_v^{\top}(\ell)] = 0$ for any $\ell \neq k$.
\end{enumerate}
\end{seclem}

\begin{proof}
Let $w_{v,i}(k)$ denote the $i$th element of $w_v(k)$. Note that each element is independent. Without loss of generality, we suppose $y_i(k) \in (nd, (n+1)d]$. Then, there exists $z \in [0, d)$ such that $y_i(k) = z + nd$.

(Proof of item (i)) 
It follows from~\eqref{static} that
\begin{align*}
\bE [w_{v,i}(k)] 
&= \bE [\mathcal{Q}_v(z) - (z + nd)] \\
&= \bE [\mathcal{Q}_v(z)] - (z + nd) \\
&= \left( nd \left(1 - \frac{z}{d}\right) + (n+1)d \frac{z}{d} \right)  - (z + nd)
= 0.
\end{align*}
This completes the proof.

(Proof of item (ii)) 
It follows from~\eqref{static} that
\begin{align*}
\bE [w_{v,i}^2(k)] 
&= \bE [(\mathcal{Q}_v(z) - (z + nd))^2] \\
&= \bE [(\mathcal{Q}_v(z))^2] - (z + nd)^2 \\
&= (nd)^2 \left(1 - \frac{z}{d}\right) + ((n+1)d)^2 \frac{z}{d}  - (z + nd)^2\\
&= - z^2 + d z
=- \left( z - \frac{d}{2}\right)^2 + \frac{d^2}{4}
\le \frac{d^2}{4}
\end{align*}
for all $z \in [0, d)$. Since $w_{v,i}(k)$ and $w_{v,j}(k)$, $j \neq i$ are independent, we have by item (i) that
\begin{align*}
\bE [w_{v,i}(k) w_{v,j}(k)] = \bE [w_{v,i}(k)] \bE[w_{v,j}(k)] = 0.
\end{align*}
That completes the proof.

(Proof of item (iii)) This follows  from  the fact that $w_v(k)$ and $w_v(\ell)$, $k \neq \ell$ are independent.
\end{proof}

Now, we are ready to provide the proof of Theorem~\ref{thm:J}.

({\it Proof of Theorem~\ref{thm:J}})
(Step 1)
Let $\bar{x} := \hat{x} - X x_r$. Then, it follows from Assumption~\ref{asm:4} and $e_x = \hat x -x$ that 
\begin{align*}
e_y (k) 
&= H_p x(k) - H_r x_r(k) \\
&= H_p x(k) - H_p X x_r(k) \\
&= H_p x(k) - H_p (\hat{x}(k) - \bar x (k) )  \\
&= H_p \bar x (k) - H_p e_x(k)
= H_p
\begin{bmatrix}
I \\ - I
\end{bmatrix}^\top
\begin{bmatrix}
\bar x (k)\\ e_x(k)
\end{bmatrix}.
\end{align*}
Noting that taking the expectation and trace is commutative, we have
\begin{align}\label{pf1:J}
    &\bE [e_y ^{\top}(k) Q e_y (k)] \nonumber\\
    &= 
\bE \left[
\begin{bmatrix}
\bar x (k)\\ e_x(k)
\end{bmatrix}^\top    
\begin{bmatrix}
I \\ - I
\end{bmatrix}
H_p^\top Q  H_p
\begin{bmatrix}
I \\ - I
\end{bmatrix}^\top
\begin{bmatrix}
\bar x (k)\\ e_x(k)
\end{bmatrix}\right]\nonumber\\
&={\rm trace}\left( 
\begin{bmatrix}
I \\ - I
\end{bmatrix}
H_p^\top Q  H_p
\begin{bmatrix}
I \\ - I
\end{bmatrix}^\top
\bE\left[
\begin{bmatrix}
\bar x (k)\\ e_x(k)
\end{bmatrix}
\begin{bmatrix}
\bar x (k)\\ e_x(k)
\end{bmatrix}^\top 
\right]
\right) \nonumber\\
&\le {\rm trace}\left( 
\begin{bmatrix}
I \\ - I
\end{bmatrix}
H_p^\top Q  H_p
\begin{bmatrix}
I \\ - I
\end{bmatrix}^\top\right)\nonumber\\
&\qquad 
 {\rm trace}\left( \bE\left[
\begin{bmatrix}
\bar x (k)\\ e_x(k)
\end{bmatrix}
\begin{bmatrix}
\bar x (k)\\ e_x(k)
\end{bmatrix}^\top 
\right] \right)\nonumber\\
&=2 {\rm trace}(H_p^\top Q  H_p) {\rm trace}\left( \bE\left[
\begin{bmatrix}
\bar x (k)\\ e_x(k)
\end{bmatrix}
\begin{bmatrix}
\bar x (k)\\ e_x(k)
\end{bmatrix}^\top 
\right] \right),
\end{align}
where in the inequality, ${\rm trace}(AB) \le {\rm trace}(A) {\rm trace}(B)$ for $A, B \succeq 0$ with the same size is used .

(Step 2)
From~\eqref{eq:sys},~\eqref{reference_sys}, and Assumption~\ref{asm:4}, we have
\begin{align*}
\bar{x} (k+1) &= \hat{x}(k+1) - X x_r(k+1)\\
&= (A+BK_x)\hat{x}(k) + B K_r x_r(k) \\
&\quad - XA_r x_r(k) +  L (Cx(k)+Ce_x(k)-v(k))\\
&= (A+BK_x)\hat{x}(k) + B(U- K_x X) x_r(k) \\
&\quad - XA_r x_r(k) - L w_v(k) + LCe_x(k) \\
&= (A+BK_x) \bar{x}(k) - L w_v(k) + LCe_x(k),
\end{align*}
and
\begin{align*}
    e_x(k+1) = (A+LC) e_x(k) - L w_v(k).
\end{align*}
Namely, we obtain
\begin{align*}
    \begin{bmatrix}
        \Bar{x}(k+1) \\ e_x(k+1)
    \end{bmatrix} =  \cA  
    \begin{bmatrix}
        \Bar{x}(k) \\ e_x(k)
    \end{bmatrix} - \begin{bmatrix}
        I \\ I
    \end{bmatrix} L w_v(k),
\end{align*}
where $\cA $ is defined in~\eqref{eq:cA}.

Since $w_v(k)$ is independent of $\bar{x}(k)$ and $e_x(k)$,
\begin{align}\label{pf3:J}
P(k) := \bE\left[\begin{bmatrix}
      \bar{x}(k) \\ e_x(k)
    \end{bmatrix}
    \begin{bmatrix}
      \bar{x}(k) \\ e_x(k)
    \end{bmatrix}^\top \right]
\end{align}
satisfies
\begin{align*}
P(k+1) 
= \cA  P(k) \cA ^{\top} 
+ 
\begin{bmatrix}
        I \\ I
    \end{bmatrix} L \bE[ w_v(k) w_v^{\top}(k)] L^\top \begin{bmatrix}
        I \\ I
    \end{bmatrix}^\top.
\end{align*}
From item (ii) of Lemma~\ref{lem:2.1}, i.e., $\bE[w_v(k)w_v^{\top}(k)] \leq \frac{d^2}{4} I $, we have
\begin{align*}
P(k+1) 
\preceq \cA  P(k) \cA ^{\top} 
+ \frac{d^2}{4}
\begin{bmatrix}
        I \\ I
    \end{bmatrix}
    L L^\top
    \begin{bmatrix}
        I \\ I
    \end{bmatrix}^\top,
\end{align*}
and consequently
\begin{align}\label{pf4:J}
P(k) 
\preceq
\cA ^k P(0) (\cA ^k)^\top
+\frac{d^2}{4} \sum_{\ell = 0}^{k-1} \cA ^\ell \begin{bmatrix}
        I \\ I
    \end{bmatrix}
    L L^\top
    \begin{bmatrix}
        I \\ I
    \end{bmatrix}^\top (\cA ^\ell)^\top.
\end{align}

(Step 3)
Recall that $A+BK_x$ and $A + LC$ are Schur stable. Thus, the Lyapunov equation~\eqref{eq:Lya} has the symmetric and positive semi-definite solution $Z$:
\begin{align*}
Z = \sum_{\ell = 0}^\infty \cA ^\ell \begin{bmatrix}
        I \\ I
    \end{bmatrix}
    L L^\top
    \begin{bmatrix}
        I \\ I
    \end{bmatrix}^\top (\cA ^\ell)^\top.
\end{align*}
Thus, $P(k)$ in~\eqref{pf4:J} can be upper bounded as
\begin{align*}
P(k) \preceq \cA ^k P(0) (\cA ^k)^\top +\frac{d^2}{4} Z,
\quad \forall k \in \bZ_+,
\end{align*}
and thus
\begin{align}\label{pf6:J}
\lim_{k \to \infty} P(k) 
\preceq 
\lim_{k \to \infty} \cA ^k P(0) (\cA ^k)^\top + \frac{d^2}{4} Z
=\frac{d^2}{4} Z.
\end{align}

Finally, it follows from~\eqref{pf1:J},~\eqref{pf3:J}, and~\eqref{pf6:J} that
\begin{align*}
\lim_{k \to \infty} \bE [e_y ^{\top}(k) Q e_y (k)] \le \frac{d^2}{2} {\rm trace}(H_p^\top Q  H_p) {\rm trace} (Z).
\end{align*}
From~\eqref{J}, this is nothing but~\eqref{eq:Jub}.
\QED



\section{Proof of Theorem~\ref{thm:dyn_quantizer_arbi}}\label{app3}
Combining \eqref{pf1.5:dp} -- \eqref{pf3:dp} leads to
\begin{align*}
&\sup_{|x_0-x'_0|_1 \leq \zeta }\left|\bP \left(V_k \in \cS \mid x_0\right) - \bP \left(V'_k\in \cS \mid x'_0\right)\right| \\
&\leq \sum_{t = 0}^{k} \sup _{|y(t)-y(t)^{\prime}|_1 \leq \beta |C|_1 \lambda^{t} \zeta }  | \bP \left(v(t)\in \cS_t \mid x_0\right) \\
&\hspace{40mm}- \bP \left(v'(t)\in \cS_t \mid x'_0\right) |.
\end{align*}
Following similar proceedure as~\eqref{pf4:dp} with~\eqref{constraint_quantizer}, we have
\begin{align*}
 &\sup _{|x_0-x'_0|_1 \le \zeta }\left|\bP \left(V_k \in \cS \mid x_0\right) - \bP \left(V'_k\in \cS \mid x'_0\right)\right| \\
 &\leq  \sum_{t = 0}^{k}     \frac{\beta |C|_1 \lambda^{t} \zeta}{d^{*} + (d(0)-d^{*})q^{t}} \\
 &\leq  \sum_{t = 0}^{k}     \frac{\beta |C|_1 \lambda^{t} \zeta}{\max \{ d^{*}, d(0)q^{t}\}} \\
  &\leq  \sum_{t = 0}^{k}    \min \left\{ \frac{\beta |C|_1 \lambda^{t} \zeta}{ d(0)q^{t}}, \frac{\beta |C|_1 \lambda^{t} \zeta}{ d^*} \right\}
  \leq \delta,
\end{align*}
where the second inequality holds since $d^{*} + (d(0)-d^{*})q^{t} \geq \max \{d^{*}, d(0)q^t \}$.
\QED

\section{Proof of Theorem~\ref{thm:input_noise}}\label{app4}
Before proving the theorem, we first provide the following lemma.
\begin{seclem}
\label{lem:dp}
    For any $\varepsilon \geq 0$, $\delta \in(0,1)$, and $\zeta > 0$, the Gaussian mechanism $M(x)=F x+W$ with $W \sim \mathcal{N}\left(0, \mathbf{I}_n\right)$ is $(\varepsilon, \delta)$-differentially private for $\operatorname{Adj}_1^\zeta$ if
\begin{align*}
1 \ge \frac{|F|_2 \zeta }{\kappa_\varepsilon^{-1}(\delta)}.
\end{align*}
\end{seclem}
\begin{proof}
    The proof is similar to that in \cite[Lemma 1]{wang2023differential} and hence is omitted.
\end{proof}

Now, we are ready to prove Theorem~\ref{thm:input_noise}.

(Proof of Theorem~\ref{thm:input_noise})
Since $CA^{k}B = 0 $ for all $0 \leq k \leq n^{*}-2$, we have $N_k = 0$ from~\eqref{eq:Nk}, and thus the mechanism \eqref{eq:mech2} satisfies
\begin{align*}
V_k = \mathcal{Q}_v (O_kx_0 + N_k (U_k + W_k))= \mathcal{Q}_v (O_kx_0)
\end{align*}
for $k \leq n^{*} - 1$.
Thus, it follows that
\begin{align*}
&\bP \left(V_k \in \cS_1 | x_0\right) \\ 
&= \bP(({Q}_v (O_kx_0 )\in \cS_1) \leq  \bP(({Q}_v (O_kx'_0)\in \cS_1) + \delta_1 \\
& \leq \bP\left(V'_k \in \cS_1 | x'_0\right) + \delta_1, \quad \forall k \leq n^* -1, \forall \cS_1 \in \cF,
\end{align*}
where the second inequality follows from item 2) and Theorem~\ref{thm:dyn_quantizer_arbi}.
From item 1), the mechanism is $(\varepsilon, \delta)$ differentially private for any $k \leq n^{*} -1$.

Next, we consider the case where $k \geq n^{*}$. 
From the definition of $M$, $x(n^*)$ can be described as
    \begin{align*}
        x(n^*) = A^{n^*} x(0) + MU_{n^* -1} + MW_{n^*-1}.
    \end{align*}
 Since $\Delta \succ 0$, we have
    \begin{align*}
&\frac{1}{\sigma}\Delta^{-\frac{1}{2}}x(n^*) \\
&= \frac{1}{\sigma}(\Delta^{-\frac{1}{2}}A^{n^*} x(0) + \Delta^{-\frac{1}{2}}MU_{n^* -1} + \Delta^{-\frac{1}{2}}MW_{n^*-1}).
    \end{align*}
This is nothing but the transforms of the covariance of Gaussian noise to $I_{n_1}$. 
Since $U_{n^*-1}$ is public information and $|x_0 - x'_0|_1 \leq \zeta$, by Lemma~\ref{lem:dp} $(\varepsilon_0, \delta_2)$ differential privacy is achieved if 
\begin{align*}
\sigma \geq \frac{| \Delta^{-\frac{1}{2}}A^{n^*}|_2\zeta}{k_{\varepsilon_0}^{-1}(\delta_2)}.
\end{align*}
Moreover, since $x(n^{*})$ is differentially private under posterior operation \cite[Theorem 1]{Le2013} and $V_{n^{*}:k} := \begin{bmatrix}
    v(n^{*}) & v(n^{*}+ 1)  & \cdots & v(k) 
\end{bmatrix}$ is a posterior operation of $x(n^*)$, we have 
\begin{align*}
\bP \left(V_{n^{*}:k} \in \cS_2 | x_0\right) \leq e^{\varepsilon_0}\bP \left(V^{\prime}_{n^{*}:k}  \in \cS_2| x'_0 \right)+\delta_2, 
\quad \forall \cS_2 \in \cF.
\end{align*}
From the independence of $\mathcal{Q}_v$ and $W_{n^{*}-1}$, we have
\begin{align*}
&\bP \left(V_{k} \in \cS | x_0\right) \\
&= \bP \left(V_{n^*-1} \in \cS_1 | x_0 \right) \bP \left(V_{n^*:k} \in \cS_2 | x_0\right)\\
&\le  (\bP (V'_{n^*-1} \in \cS_1 | x'_0 ) + \delta_1) \bP \left(V_{n^*:k} \in \cS_2 | x_0\right)\\ 
&\le \bP \left(V'_{n^*-1} \in \cS_1 | x'_0 \right) \bP \left(V_{n^*:k} \in \cS_2 | x_0\right) + \delta_1 \\
&\le \bP \left(V'_{n^*-1} \in \cS_1 | x'_0 \right) \left(e^{\varepsilon_0}\bP \left(V^{\prime}_{n^*:k} \in \cS_2 | x'_0 \right) + \delta_2\right) + \delta_1\\
&\le e^{\varepsilon_0} \bP \left(V'_{n^*-1} \in \cS_1 | x'_0 \right) \bP \left(V^{\prime}_{n^*:k} \in \cS_2 | x'_0 \right) + \delta_2 + \delta_1\\
&\le  e^{\varepsilon}\bP \left(V'_k  \in \cS| x'_0 \right)+\delta,
\end{align*}
where the last inequality follows from item 1). 
\QED

\bibliographystyle{ieeetr}
\bibliography{ref}

\end{document}